\documentclass[12pt,letterpaper]{article}
\usepackage{amsmath, amssymb, amsthm, amsfonts, mathrsfs,bbm}
\usepackage{cite,algorithmic,algorithm, psfrag}
\usepackage{graphicx, subcaption}
\usepackage{tikz-cd}
\usepackage[dvips,letterpaper,margin=1in]{geometry}

\newtheorem{thm}{Theorem}
\newtheorem{lm}{Lemma}

\newcommand{\BEAS}{\begin{eqnarray*}}
\newcommand{\EEAS}{\end{eqnarray*}}
\newcommand{\BEQ}{\begin{equation}}
\newcommand{\EEQ}{\end{equation}}
\newcommand{\BIT}{\begin{itemize}}
\newcommand{\EIT}{\end{itemize}}
\newcommand{\ie}{{\it i.e.}}
\newcommand{\intr}{\mathop{\bf int}}
\newcommand{\tr}{\mathbf{Tr}}
\newcommand{\symm}{\mathbf{S}}

\newcommand{\reals}{{\mbox{\bf R}}}

\newcommand{\var}{{\mathbf{Var}}}
\newcommand{\expec}{\mathbf{E}}

\title{Adaptive Importance Sampling\\
via Stochastic Convex Programming}
\usepackage[blocks]{authblk}
\author[1]{Ernest K. Ryu}
\author[1]{Stephen P. Boyd}
\affil[1]{Institute for Computational and Mathematical Engineering, Stanford University}
\date{\today}

\begin{document}
\maketitle
\begin{abstract}
We show that the variance of the Monte Carlo estimator
that is importance sampled from an
exponential family is a convex function of the natural parameter of the distribution.
With this insight, we propose an adaptive importance sampling algorithm
that simultaneously improves the choice of sampling distribution while accumulating a
Monte Carlo estimate.
Exploiting convexity, we  prove that the method's unbiased estimator
has variance that is asymptotically optimal over the exponential family.
\end{abstract}

\noindent


\section{Introduction}
Consider the problem of approximating the expected value (or integral) 
\[
I=\mathbf{E} \phi(X)=\int \phi(x)f(x)\; dx,
\]
where $X\sim f$ is a random variable on $\reals^k$
and $\phi:\reals^k\rightarrow \reals$.

The standard Monte Carlo method estimates $I$ by taking independent identically 
distributed (IID) samples $X_1,X_2,\ldots, X_n\sim f$ and using
\[
\hat{I}_n^\mathrm{MC}=\frac{1}{n}\sum^n_{i=1}\phi(X_i).
\]
This estimator is \emph{unbiased}, \ie,
$\expec \hat{I}^\mathrm{MC}_n=I$,
and has variance
\[
\var( \hat{I}^\mathrm{MC}_n)=\frac{1}{n}
\var_{X\sim f}\left[
\phi(X)\right]
=
\frac{1}{n}
\left(
\int \phi^2(x)f(x)\;dx-I^2\right).
\]

To reduce the variance of the estimator, the standard figure of merit,
one can use importance sampling:
choose a \emph{sampling} (\emph{importance}) distribution
$\tilde{f}$ satisfying $\tilde{f}(x)>0$ whenever $\phi(x)f(x)\ne 0$,
take IID samples $X_1,X_2,\ldots,X_n\sim \tilde{f}$
(as opposed to sampling from $f$, the \emph{nominal} distribution)
and use
\[
\hat{I}^{\mathrm{IS}}_n=\frac{1}{n}\sum^n_{i=1}\phi(X_i)\frac{f(X_i)}{\tilde{f}(X_i)}
.
\]
Again, the estimator is unbaised,
\ie, $\expec \hat{I}^\mathrm{IS}_n=I$,
and has variance
\[
\var( \hat{I}^\mathrm{IS}_n)=\frac{1}{n}
\var_{X\sim \tilde{f}}\left[\frac{\phi(X)f(X)}{\tilde{f}(X)}\right]
=\frac{1}{n}
\left(
\int 
\frac{\phi^2(x)f^2(x)}{\tilde{f}(x)}\;dx
-I^2
\right)
.
\]
When $\tilde{f}=f$, importance sampling reduces to standard Monte Carlo.
Choosing $\tilde{f}$ wisely can reduce the variance, 
but this can be difficult in general.
One approach is to use a priori information on the integrand $\phi(x)f(x)$
to manually find an appropriate sampling distribution $\tilde{f}$,
perhaps through several informal iterations \cite{sadowsky1990, smith1997, glasserman1999, srinivasan2002, madras2002, robert2004, owen2013}.
Another approach is to automate the process of finding the sampling distribution
through adaptive importance sampling.

In adaptive importance sampling, one adaptively improves the sampling distribution
while simultaneously accumulating the estimate for $I$.
A particular form of importance sampling
generates a sequence of sampling distributions
$\tilde{f}_1,\tilde{f}_2,\ldots$ 
and a series of samples $X_1\sim \tilde{f}_{1},\,X_2\sim \tilde{f}_2,\,\ldots$
and forms the estimate
\[
\hat{I}_n^\mathrm{AIS} = \frac{1}{n}\sum^n_{i=1}\phi(X_i)\frac{f(X_i)}{\tilde{f}_i(X_i)}.
\]
At each iteration $n$, the sampling distribution $\tilde{f}_n$, which is itself random, is adaptively
determined based on the past data,
$\tilde{f}_n,\ldots,\tilde{f}_{n-1}$ and $X_1,\ldots,X_{n-1}$.
Again, $\hat{I}_n^\mathrm{AIS}$ is unbiased, \ie, $\expec \hat{I}_n^\mathrm{AIS}=I$, and
\[
\var(\hat{I}_n^\mathrm{AIS}) 
=
\frac{1}{n^2}
\sum^n_{i=1}
\expec_{\tilde{f}_i}
\var_{X_i\sim \tilde{f}_i}\left[\frac{\phi(X_i)f(X_i)}{\tilde{f}_i(X_i)}\right]
=
\frac{1}{n^2}
\sum^n_{i=1}
\expec_{\tilde{f}_i}
\left(
\int 
\frac{\phi^2(x)f^2(x)}{\tilde{f}^2(x)}\;dx
-I^2
\right),
\]
where $\expec_{\tilde{f}_i}$ denotes the expectation over the random sampling distribution $\tilde{f}_i$.
Again, when $\tilde{f}_i=\tilde{f}$ for all $i$, adaptive importance sampling reduces to
standard (non-adaptive) importance sampling.
Now determining how to choose $\tilde{f}_n$ at each iteration fully specifies the method.

In this paper, we propose an instance of adaptive importance sampling,
which we call Convex Adaptive Monte Carlo (\textsc{Convex AdaMC}).
First, we choose an exponential family of distributions $\mathcal{F}$
as the set of candidate sampling distributions.
Define 
$T:\reals^k\rightarrow \reals^p$ and $h:\reals^k\rightarrow \reals_+$.
Then our density function is
\[
f_\theta(x)=\exp \left(\theta^TT(x)-A(\theta)\right)h(x),
\]
where $A:\reals^p\rightarrow \reals\cup\{\infty\}$, defined as
\[
A(\theta)
=
\log \int \exp(\theta^TT(x))h(x)\;dx,
\]
serves as a normalizing factor.
(When $A(\theta)=\infty$, we define $f_\theta=0$ and remember that this does not define a distribution.)
Finally, let $\Theta\subseteq \reals^p$ be a convex set,
and our exponential family is
$\mathcal{F}=\left\{f_\theta\,\mid\,\theta\in \Theta\right\}$,
where $\theta$ is called the \emph{natural parameter} of $\mathcal{F}$.
Note that the choice of $T$, $h$, and $\Theta$ fully specifies our family $\mathcal{F}$.

Next, define $V:\reals^p\rightarrow \reals\cup\{\infty\}$ to be the per-sample
variance of the
importance sampled estimator with sampling distribution $f_\theta$,
\[
V(\theta)=\var_{X\sim f_\theta}
\left[\frac{\phi(X)f(X)}{f_\theta(X)}\right]
=\int \frac{\phi^2(x)f^2(x)}{f_\theta(x)}\;dx-I^2.
\]
(So the importance sampled estimator using $n$ IID samples from $f_\theta$
has variance $V(\theta)/n$.)
A natural approach is to solve
\begin{equation}
\begin{array}{ll}
\mbox{minimize}& V(\theta)\\
\mbox{subject to}& \theta\in \Theta,
\end{array}
\label{p-min-v}
\end{equation}
where $\theta$ is the optimization variable, as this will give us the best
sampling distribution among $\mathcal{F}$ to importance sample from.
We write $V^\star$ to denote the optimal value, \ie, the optimal per-sample variance
over the family.

The first key insight of this paper is that $V$ is a convex function,
a consequence of $\mathcal{F}$ being an exponential family.
Roughly speaking, one can efficiently find a \emph{global} minimum of a convex functions
through standard methods if one can compute the function value and its gradient
\cite{nesterov2004}.
This fact, however, is not directly applicable to our setting as
evaluating $V(\theta)$ or $\nabla V(\theta)$ for any given $\theta$
is in general as hard as evaluating $I$ itself.

The second key insight is that we can minimize
the convex function $V$
through a standard algorithm of stochastic optimization,
\emph{stochastic gradient descent},
while simuiltaneously accumulating an estimate for $I$.
Because of convexity, we can prove theoretical guarantees.

In \textsc{Convex AdaMC}, we generate a sequence of sampling distribution parameters
$\theta_1,\theta_2,\ldots$ 
and a series of samples $X_1\sim f_{\theta_1},\,X_2\sim f_{\theta_2},\ldots$,
with which we form the estimate
\[
\hat{I}_n^\mathrm{AMC} = \frac{1}{n}\sum^n_{i=1}\phi(X_i)\frac{f(X_i)}{f_{\theta_i}(X_i)}.
\]
Again, $\hat{I}_n^\mathrm{AMC}$ is unbiased, \ie, $\expec \hat{I}_n^\mathrm{AMC}=I$.
Furthermore, we show that
\begin{equation}
\var(\hat{I}_n^\mathrm{AMC}) 
=\frac{1}{n}V^\star + \mathcal{O}\left(\frac{1}{n^{3/2}}\right) =
\frac{1}{n}\left( V^\star + \mathcal{O}\left(\frac{1}{n^{1/2}}\right)\right).
\label{e-asymp-var}
\end{equation}
This shows that the per-sample variance of \textsc{Convex AdaMC}
converges to the optimal per sample variance over our family $\mathcal{F}$;
\ie,
our estimator \textsc{Convex AdaMC} asymptotically performs as well as
any (adaptive or non-adaptive) importance sampling estimator using
sampling distributions from $\mathcal{F}$.
In particular, \textsc{Convex AdaMC} does not suffer from 
becoming trapped in (non-optimal) local minima,
a problem other adaptive importance sampling methods can have.

\section{Convexity of the variance}
\label{s-variance}
Let's establish a few important properties of our variance function
\[
V(\theta)=\int \phi^2(x)f^2(x)\exp(A(\theta)-\theta^T T(x))h(x)\;dx
-I^2.
\]
When $A(\theta)=\infty$, we define $V(\theta)=\infty$.
Not only is this definition natural but is also convenient
since $V(\theta)=\infty$ now indicates that $\theta$ is invalid
either because the variance is infinite
or because $\theta$ doesn't define a sampling distribution.

Recall that a function 
$V:\reals^p \rightarrow\reals \cup\{\infty\}$ is convex if
\[
V(\eta \theta_1+(1-\eta) \theta_2)\le \eta V(\theta_1)+(1-\eta)V(\theta_2)
\]
holds for any $\eta\in[0,1]$ and $\theta_1,\theta_2\in \Theta$.
Convexity is important because it allows us to prove a theoretical guarantee.
\begin{thm}
The variance of the importance sampling estimator $V(\theta)$
is a convex function of $\theta$, the natural parameter of the exponential family.
\end{thm}
\begin{proof}
We first show $A(\theta)$ is convex.
By H\"older's inequality, we have
\begin{align*}
\exp(A(\eta\theta_1+(1-\eta)\theta_2))&
=
\int \exp((\eta\theta_1+(1-\eta)\theta_2)^T T(x))h(x)\;dx\\
&\le
\left(\int \exp(\theta_1^T T(x))h(x)\;dx\right)^\eta
\left(\int \exp(\theta_2^T T(x))h(x)\;dx\right)^{1-\eta},
\end{align*}
and by taking the log on both sides we get
\[\\
A(\eta\theta_1+(1-\eta)\theta_2)
\le
\eta A(\theta_1)
+(1-\eta) A(\theta_2).
\]

Since $\exp(\cdot)$ is an increasing convex function and $A(\theta)-\theta^TT(x)$ is convex in $\theta$,
the composition $\exp(A(\theta)-\theta^T T(x))$ is convex in $\theta$.
Finally, $V(\theta)$ is convex as it is
an integral of the convex functions
$\exp(A(\theta)-\theta^T T(x))h(x)$ over $x$;
see
\cite[\S B.2]{hiriart2001},
\cite[\S5]{rockafellar1970},
or \cite[\S3.2]{boyd2004}.
\end{proof}
We note in passing that $\log V(\theta)$ is also a convex
function of $\theta$, which is
a stronger statement than convexity of $V(\theta)$.
This fact, however, is not useful for us
since we do not have a simple way to obtain a 
stochastic gradient for $\log V(\theta)$,
whereas, as we will see later, we do for $V(\theta)$.

As we will see soon,
stochastic gradient descent hinges on evaluating the derivative of $V$ under the integral.
The following lemma is a consequence of Theorem 2.7.1 of \cite{lehmann2005}.
\begin{lm}
\label{lm-diff}
$V$ is differentiable and its gradient can be evaluated
under the integral on
$\intr\left\{\theta \,\mid\,V(\theta)<\infty\right\}$,
where $\intr$ denotes the interior.
\end{lm}
In particular, we have
\BEAS
\nabla V(\theta)&=&
\int\nabla_\theta\frac{\phi^2(x)f^2(x)}{f_\theta(x)}\;dx\\
&=&
\int(\nabla A(\theta)-T(x))\frac{\phi^2(x)f^2(x)}{f_\theta^2(x)}f_\theta(x)\;dx\\
&=&
\expec_{X\sim f_\theta}
\left[
(\nabla A(\theta)-T(X))\frac{\phi^2(X)f^2(X)}{f_\theta^2(X)}
\right].
\EEAS
So when we take a sample $X\sim f_{\theta}$,
the random vector
\[
g=(\nabla A(\theta)-T(X))\frac{\phi^2(X)f^2(X)}{f_\theta^2(X)}
\]
satisfies $\expec g=\nabla V(\theta)$.

\section{The method}
\label{s-method}
\emph{Stochastic gradient descent} is a standard method for solving
\[
\begin{array}{ll}
\mbox{minimize}&V(\theta)\\
\mbox{subject to}&\theta\in \Theta,
\end{array}
\]
using the algorithm
\[
\theta_{n+1}=\Pi(\theta_n-\alpha_ng_n),
\]
where $\Pi$ is (Euclidean) projection onto $\Theta$,
the \emph{step size} $\alpha_n>0$ is an appropriately chosen sequence,
and the \emph{stochastic gradient} $g_n$ is a random variable satisfying
\[
\expec
\left[
g_n\mid\theta_n
\right]
=\nabla V(\theta_n).
\]
The intuition is that $-g_n$, although noisy, generally points towards
a descent direction of $V$ at $\theta_n$, and therefore each step
reduces the function value of $V$ in expectation
\cite{robbins1951, shor1985, polyak1987, kushner2003}.

Our algorithm, which we call \textsc{Convex AdaMC}, is
\BEAS
X_n&\sim& f_{\theta_n}\\
\hat{I}^\mathrm{AMC}_n&=&\frac{1}{n}\sum^n_{i=1}\frac{\phi(X_i)f(X_i)}{f_{\theta_i}(X_i)}\\
g_n&=&(\nabla A(\theta_n)-T(X_n))\frac{\phi^2(X_n)f^2(X_n)}{f_{\theta_n}^2(X_n)}\\
\theta_{n+1}&=&\Pi\left(\theta_n-\frac{C}{\sqrt{n}}g_n\right),
\EEAS
where $C>0$ and $\theta_1\in \Theta$.
As mentioned in the introduction,
the estimator $\hat{I}^\mathrm{AMC}_n$ is unbiased and has variance given by
\eqref{e-asymp-var} under a technical condition to be presented in \S\ref{s-analysis}.

We can view \textsc{Convex AdaMC}
as an adaptive importance sampling method
where the third and fourth line of the algorithm updates the sampling distribution.
Alternatively, we can view \textsc{Convex AdaMC} as
stochastic gradient descent on the convex function $V$
with an additional step, the second line of the algorithm,
that accumulates the estimate of $I$
but does not otherwise affect the iteration.

The computational cost of \textsc{Convex AdaMC} is cheap, of course,
if all of its operations are cheap.
This is the case if
$\phi$ and $f$ are functions we can easily evaluate,
if our family of distributions, $\mathcal{F}$, is one of the well-known exponential families,
and if $\Theta$ is a set we can easily project onto.

\section{Analysis}
\label{s-analysis}
Before we present our convergence results, we discuss the choice of $\Theta$.
For any convex domain $\Theta$, our variance function $V$ is convex
and minimizing $V(\theta)$ over $\theta\in\Theta$ is a mathematically
well-defined problem.
However, for our method to be well-defined and for the proof
of convergence to work out, we need further restrictions on $\Theta$.

Define
\[
K(\theta)=\expec_{X\sim f_\theta}
\left[
\frac{\phi^4(X)f^4(X)}{f_\theta^4(X)}
\right]
=\int 
\frac{\phi^4(x)f^4(x)}{f_\theta^3(x)}\;dx.
\]
We require that $\Theta$ is convex and compact and that
$\Theta\subseteq \intr\left\{\theta\mid K(\theta)<\infty\right\}$.
In other words, $\Theta$ must be a convex compact subset of
the interior of the set of
natural parameters for which
their importance sampled estimates
have finite 4th moment.
Since
\[
\left\{\theta\mid K(\theta)<\infty\right\}
\subseteq
\left\{\theta\mid V(\theta)<\infty\right\}
\subseteq
\left\{\theta\mid A(\theta)<\infty\right\}
\]
it follows that any $\theta\in \Theta$
defines a sampling distribution $f_\theta$
that produces an importance sampled estimate of finite variance.

\begin{thm}
\label{thm-main}
Assume 
$\Theta\subseteq \intr\left\{\theta\mid K(\theta)<\infty\right\}$ is nonempty, convex, and compact.
Define
$D=\max_{\theta_1,\theta_2\in \Theta}\|\theta_1-\theta_2\|_2$
and
\[
G^2= \sup_{\theta\in \Theta}\expec_{X\sim f_\theta}
\left\|(\nabla A(\theta)-T(X))\frac{\phi^2(X)f^2(X)}{f_\theta^2(X)}\right\|_2^2.
\]
Then $G<\infty$, and $\hat{I}^\mathrm{AMC}_n$, the unbiased estimator of \textsc{Convex AdaMC},
satisfies
\[
\frac{1}{n}V^*\le
\var(\hat{I}^\mathrm{AMC}_n)\le \frac{1}{n}V^*+\left(\frac{D^2}{2C}+CG^2\right)\frac{1}{n^{3/2}}.
\]
\end{thm}
\begin{proof}
We defer the proof of $G<\infty$ to the appendix.

Since the conditional dependency of our sequences $\theta_1,\theta_2,\ldots$ and $X_1,X_2,\ldots$ is
\[
\begin{tikzcd}
\theta_1 \arrow{d}\arrow{r} & \theta_2 \arrow{d}\arrow{r} & \theta_3 \arrow{d}\arrow{r} &\theta_4\\
X_1 \arrow{ur} & X_2 \arrow{ur} & X_3 \arrow{ur} &
\end{tikzcd}
\cdots
\]
$X_i$ is independent of the entire past conditioned on $\theta_i$ for all $i$.
With this insight, we have
\BEAS
\expec \hat{I}^\mathrm{AMC}_n&=&\frac{1}{n}\sum^n_{i=1}\expec\frac{\phi(X_i)f(X_i)}{f_{\theta_i}(X_i)}
=\frac{1}{n}\sum^n_{i=1}\expec\left[\expec\left[\frac{\phi(X_i)f(X_i)}{f_{\theta_i}(X_i)}\mid \theta_i\right]\right]
=\frac{1}{n}\sum^n_{i=1}\expec\left[I\right]=I
\EEAS
and 
\BEAS
\var (\hat{I}^\mathrm{AMC}_n)
&=&
\expec (\hat{I}^\mathrm{AMC}_n-I)^2\\
&=&\frac{1}{n^2}
\sum^n_{i=1}
\expec \left(\frac{\phi(X_i)f(X_i)}{f_{\theta_i}(X_i)}-I\right)^2\\
&&+
\frac{2}{n^2}\sum_{1\le i<j\le n}
\expec \left(\frac{\phi(X_i)f(X_i)}{f_{\theta_i}(X_i)}-I\right)
\left(\frac{\phi(X_j)f(X_j)}{f_{\theta_j}(X_j)}-I\right)
\\
&=&
\frac{1}{n^2}
\sum^n_{i=1}\expec\left[
\expec\left[ \left(\frac{\phi(X_i)f(X_i)}{f_{\theta_i}(X_i)}-I\right)^2\mid\theta_i\right]\right]\\
&&+
\frac{2}{n^2}\sum_{1\le i<j\le n}
\expec\left[\expec\left[ \left(\frac{\phi(X_i)f(X_i)}{f_{\theta_i}(X_i)}-I\right)
\left(\frac{\phi(X_j)f(X_j)}{f_{\theta_j}(X_j)}-I\right)
\mid\theta_j\right]
\right]\\
&=&
\frac{1}{n^2}
\sum^n_{i=1}\expec V(\theta_i).
\EEAS
Since $V(\theta_i)\ge V^\star$
for any $\theta_i\in \Theta$, 
we conclude $\var (\hat{I}^\mathrm{AMC}_n)\ge V^\star/n$.\

Now let's prove the upper bound.
Let $\theta_\star$ be a minimizer of $V$ over $\Theta$ (which exists since $V$ is continuous
on the compact set $\Theta$). Then we have
\BEAS
\|\theta_{i+1}-\theta_\star\|^2_2&=& 
\|\Pi(\theta_i-C/\sqrt{i} g_i)-\Pi(\theta_\star)\|_2^2\nonumber\\
&\le&\|\theta_i-C/\sqrt{i} g_i-\theta_\star\|^2_2\\
&= &\|\theta_i-\theta_\star\|^2_2+\frac{C^2}{i}\|g_i\|^2_2-2\frac{C}{\sqrt{i}}g_i^T(\theta_i-\theta_\star),
\EEAS
where the first inequality follows from nonexpansivity of $\Pi$
(\ie, $\|\Pi(u)-\Pi(v)\|_2 \leq \|u-v\|_2$ for any $u$ and $v$).
We take expectation conditioned on $\theta_i$ on both sides to get
\BEAS
\expec\left[\|\theta_{i+1}-\theta_\star\|^2_2\mid \theta_i\right]
&\le &\|\theta_i-\theta_\star\|^2_2+\frac{C^2}{i}
\expec\left[\|g_i\|^2_2\mid \theta_i\right]-2\frac{C}{\sqrt{i}}
\nabla V(\theta_i)^T(\theta_i-\theta_\star)\\
&\le &\|\theta_i-\theta_\star\|^2_2+\frac{C^2}{i}G^2-2\frac{C}{\sqrt{i}}
\nabla V(\theta_i)^T(\theta_i-\theta_\star)\\
&\le &\|\theta_i-\theta_\star\|^2_2+\frac{C^2}{i}G^2_2-2\frac{C}{\sqrt{i}}(V(\theta_i)-V(\theta_\star)),
\EEAS
where the second inequality follows from the definition of $G$
and the third inequality follows from 
re-arranging the following consequence of $V$'s convexity
\[
V(\theta_\star)\ge V(\theta_i)+\nabla V(\theta_i)^T(\theta_\star-\theta_i).
\]
We take the full expectation on both sides and re-arrange to get
\BEAS
\expec V(\theta_i)-V^\star&\le& \frac{\sqrt{i}}{2C}(\expec\|\theta_i-\theta_\star\|_2^2-\expec\|\theta_{i+1}-\theta_\star\|_2^2)+\frac{C}{2\sqrt{i}}G^2.
\label{e-e5}
\EEAS
We take a summation to get an ``almost telescoping'' series:
\BEAS
2\sum^n_{i=1}(\expec V(\theta_i)-V^\star)&\le&
\frac{1}{C}\sum^n_{i=1}(\sqrt{i}-\sqrt{i-1})\expec\|\theta_i-\theta_\star\|_2^2+CG^2\sum^n_{i=1}\frac{1}{\sqrt{i}}\\
&\le &
\frac{D^2}{C}\sum^n_{i=1}(\sqrt{i}-\sqrt{i-1})+CG^2\sum^n_{i=1}\frac{1}{\sqrt{i}}\\
&\le &
\frac{D^2}{C}\sqrt{n}+2CG^2\sqrt{n},
\EEAS
where the second inequality follows from the definition of $D$
and the third inequality follows from
\[
\sum^n_{i=1}\frac{1}{\sqrt{i}}\le \int^n_{0}\frac{1}{\sqrt{i}}\;di.
\]
Finally, we divide both sides by $2n^2$ to get
\[
\frac{1}{n^2}\sum^n_{i=1}\expec V(\theta_i)\le
\frac{1}{n}V^\star
+
\left(\frac{D^2}{2C}+CG^2\right)\frac{1}{n^{3/2}}.
\]
\end{proof}
Not surprisingly, we have a central limit theorem (CLT) for our estimator.
The proof of the following theorem is a straightforward
application of a Martingale CLT, and is given in the appendix.
\begin{thm}
\label{thm-clt}
Under the assumptions of Theorem~\ref{thm-main}, we have
\[
\sqrt{n}(\hat{I}^\mathrm{AMC}_n-I)
\stackrel{\mathcal{D}}{\rightarrow}
\mathcal{N}(0,V^\star).
\]
\end{thm}

\textsc{Convex AdaMC} has parameters $C$ and $\theta_1$ that must be chosen,
but Theorem~\ref{thm-main} or its proof does not give us insight on how to make this choice.
Of course, the choice $C=D/\sqrt{2}G$ optimizes the bound of Theorem~\ref{thm-main},
but this is not very meaningful:
The quantity $G$, in general, is unknown a priori,
and the term $(D^2/2C+CG^2)/n^{3/2}$ is merely a bound that we suspect
is not representative of the actual performance.
In practice, one should vary $C$ through several informal iterations
to find what works well.

Likewise, the stated bound of Theorem~\ref{thm-main} independent of $\theta_1$,
and the proof does not seem to reveal any significant dependence on $\theta_1$.
However, intuition and empirical experiments suggest that a $\theta_1$ with
a small value of $V(\theta_1)$ performs well.

Rather, the theoretical significance of 
Theorem~\ref{thm-main} and \ref{thm-clt}
is that the leading order term of $\var(\hat{I}^\mathrm{AMC}_n)$ is $V^\star/n$,
the optimum among the family $\mathcal{F}$,
and that the following term is of order $\mathcal{O}(1/n^{3/2})$.
In particular, this implies that \textsc{Convex AdaMC} cannot be trapped 
at a (non-optimal) local minimum.

\section{Examples}
\label{s-examples}
\paragraph{Volume of a polytope.}
Consider the problem of computing the area of the quadrilateral $Q$
with corners at $(0.05,0.9)$, $(0.8,0.9)$, $(1,0.7)$, and $(0.15,0.7)$.
The answer is $0.16$, which of course can be found with simple geometry.

First note that
\[
I=\int^1_0\int^1_01_{Q}\;dxdy,
\]
where $1_{Q}$ is the indicator function that is $1$ within the quadrilateral
and $0$ otherwise. Now let's see how to compute $I$ with \textsc{Convex AdaMC}.

First, we choose bivariate Gaussians, which have the densities
\[
f(x;\mu,\Sigma)=
\frac{1}{2\pi|\Sigma|^{1/2}}
\exp
\left(-\frac{1}{2}(x-\mu)^T\Sigma^{-1}(x-\mu)\right),
\]
as our candidate sampling distributions.
To form these into an exponential family, we perform a change of variables.
Loosely speaking, we say our natural parameter $\theta$ has
two components: $m=\Sigma^{-1}\mu\in \reals^2$
and $S=\Sigma^{-1}\in \symm^2$ where $\symm^2$ denotes the set of $2\times 2$ symmetric matrices.
Now our densities are
\[
f_{m,S}(x)=\frac{1}{2\pi}
\exp\left(
m^Tx-\frac{1}{2}\tr(Sxx^T)
\right)
\exp\left(
-\frac{1}{2}
\left(
m^TS^{-1}m
-\log |S|
\right)
\right).
\]
(Note that $\tr(Sxx^T)$ is linear in $S$
as it is the inner product between $S$ and $xx^T$,
interpreted as vectors of $\reals^4$.)
We choose our compact natural parameter set $\Theta$ to be
\[
\Theta=
[0,25]^2
\times
\left\{S\in \symm^2\mid
I\preceq S \preceq 50I
\right\}.
\]
In other words, we restrict $m_1$ and $m_2$ to be within $[0,25]$
and the eigenvalues of $\Theta$ to both be within $[1,50]$.
With this choice, the updates of \textsc{Convex AdaMC} are
\BEAS
m_{n+1}&=&
\Pi_{[0,25]^2}\left(
m_n-\frac{C1_Q(X_n)}{f_{m,S}^2(X_n)\sqrt{n}}\left(S^{-1}_nm_n-X_n\right)\right),
\\
S_{n+1}&=&
\Pi_{\left\{S\in \symm^2\mid
I\preceq S \preceq 50I
\right\}}
\left(
S_n-\frac{C1_Q(X_n)}{2f_{m,S}^2(X_n)\sqrt{n}}
\left(
X_nX_n^T-S_n^{-1}m_nm_n^TS_n^{-1}-S_n^{-1}
\right)\right).
\EEAS

Figure~\ref{fig-examp} shows the improvement of the sampling distributions
for a particular run of this problem.
Figure~\ref{fig-examp-1} gives the initial sampling distribution.
Since the first sample to ever hit $Q$ happens at iteration $33$, 
the sampling distribution is identical for the first $32$ iterations.
As the algorithm progresses, we see that the density function of the Gaussian sampling distribution
gradually matches the shape $I_Q$.
\begin{figure}[h!]
\hskip -.35in
\begin{subfigure}{0.6\textwidth}
\includegraphics[width=\textwidth]{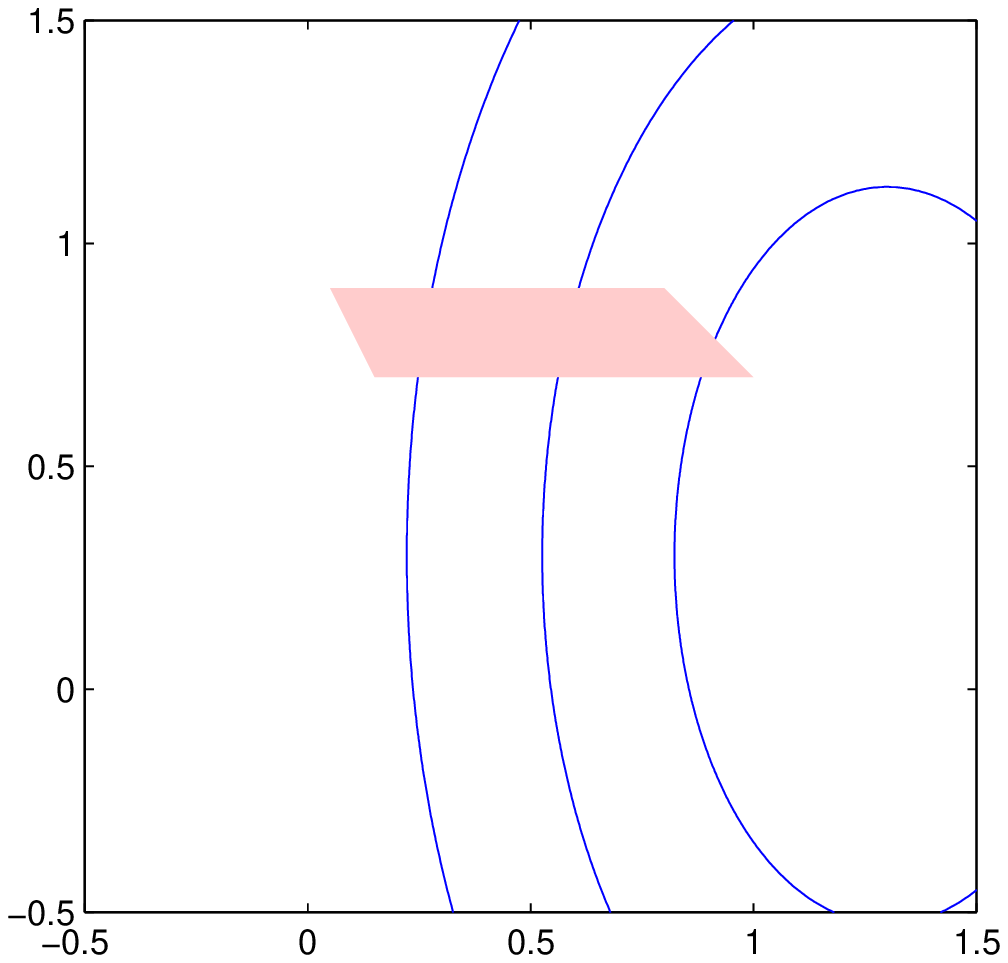}
\vskip -0.17in
\caption{Iterations $1$ through $32$}
\label{fig-examp-1}
\end{subfigure}
\hskip -.6in
\begin{subfigure}{0.6\textwidth}
\includegraphics[width=\textwidth]{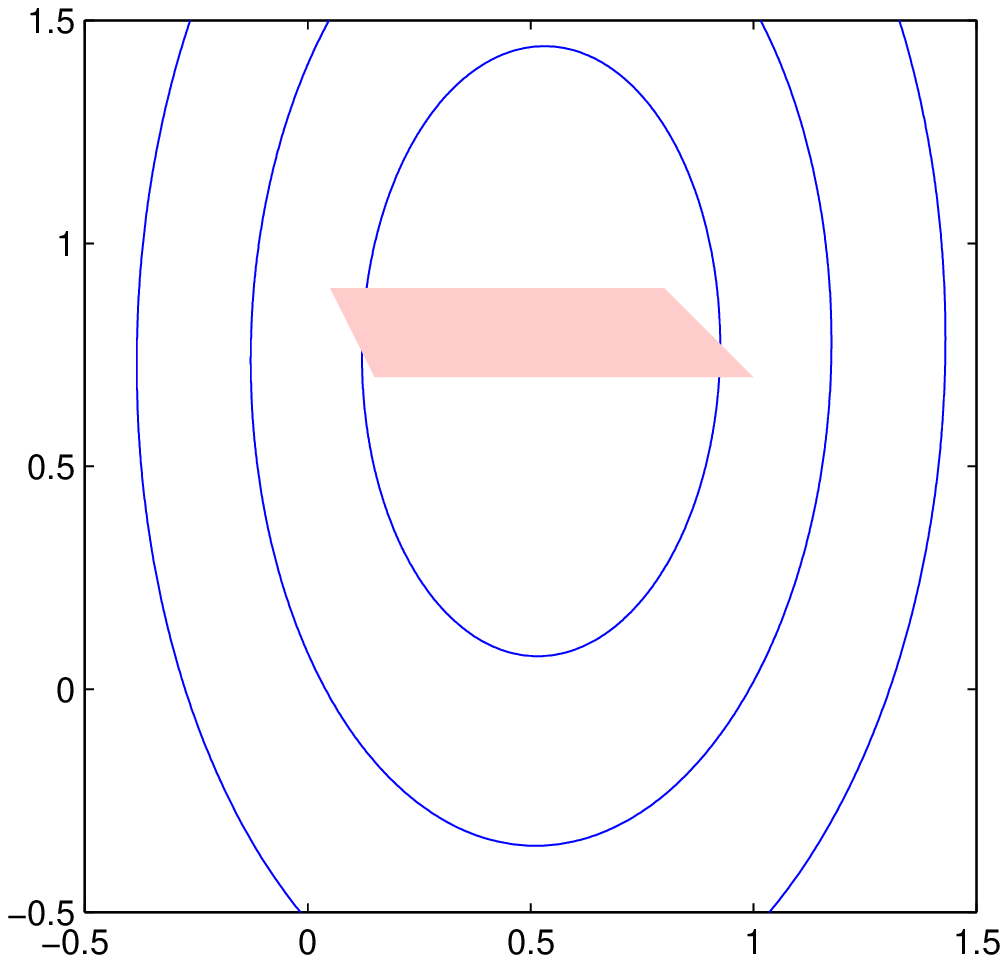}
\vskip -0.17in
\caption{Iteration $10^4$}
\end{subfigure}
\\

\hskip -.35in
\begin{subfigure}{0.6\textwidth}
\includegraphics[width=\textwidth]{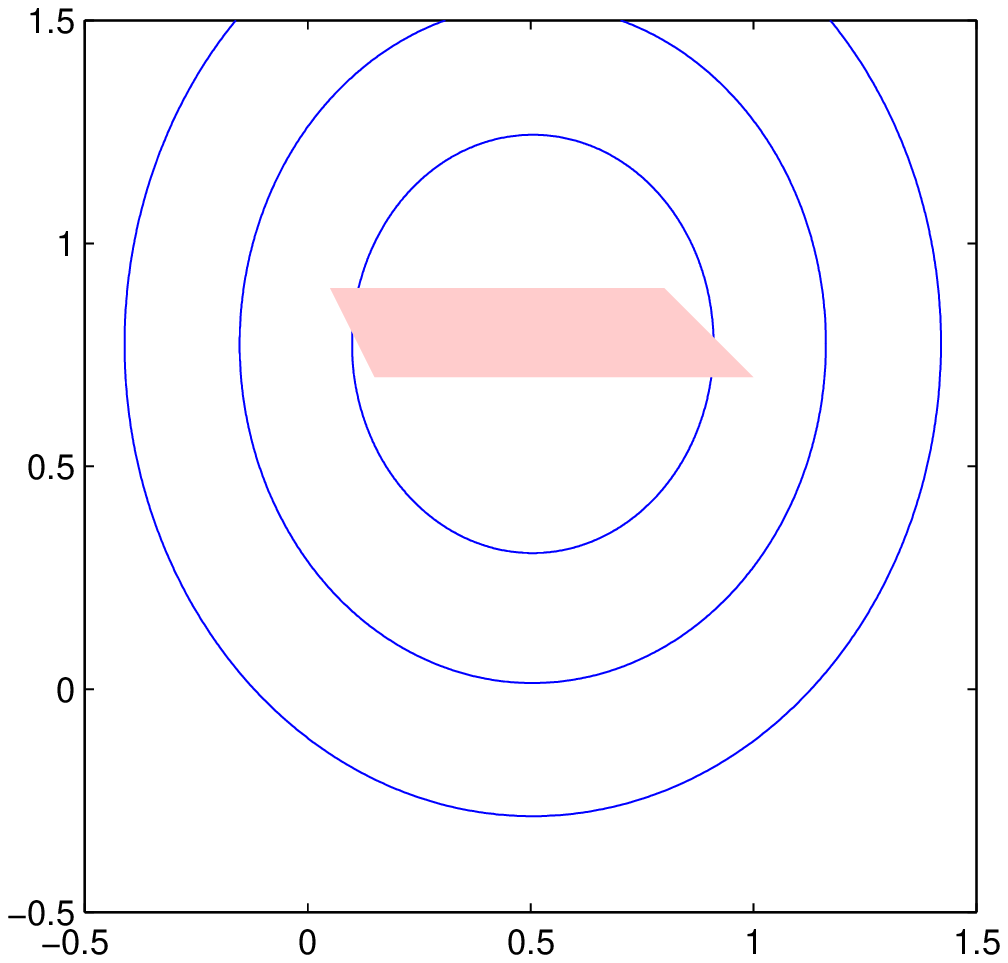}
\vskip -0.17in
\caption{Iteration $10^6$}
\end{subfigure}
\hskip -.6in
\begin{subfigure}{0.6\textwidth}
\includegraphics[width=\textwidth]{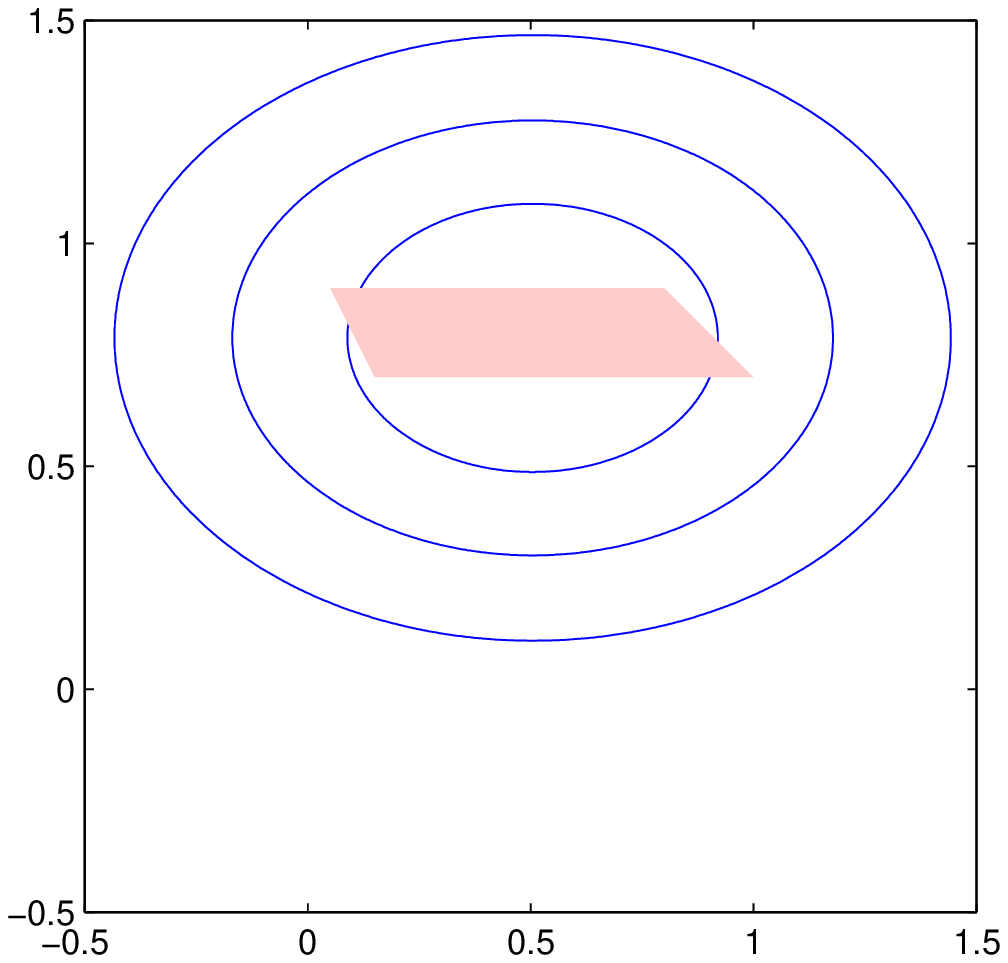}
\vskip -0.17in
\caption{Iteration $10^8$}
\end{subfigure}
\caption{Sampling distributions at different iterations.
The red quadrilateral represents $Q$ and 
the three ellipses denote the $68\%$, $95\%$, and $99.7\%$ confidence ellipsoids.}
\label{fig-examp}
\end{figure}

\paragraph{Option pricing.}
Consider the pricing of an arithmetic Asian call option
on an underlying asset under standard Black-Scholes assumptions \cite{glasserman1999}.
We write $S^0$ for the initial price of the underlying asset,
$r$ and $\sigma$ for the interest rate and volatility of the Black-Scholes model,
and $T$ for the maturity time.
Under the Black-Scholes model, the price of the asset at time $jT/k$ is
\[
S^j(X)=S^0\exp
\left[
(r-\frac{1}{2}\sigma^2)j\frac{T}{n}+\sigma\sqrt{\frac{T}{n}}\sum^j_{i=1}X^j
\right]
\]
for $t=1,\ldots,k$, where $X\in \reals^k$ is random with 
independent standard normal entries $X^1,\ldots,X^k$.
(Here we will use superscripts to denote entries of a vector.)
The discounted payoff of the option with strike $K$ is given by
\[
\phi(X)=\exp^{-rT}\max
\left\{
\frac{1}{k}\sum^k_{i=1}S^j(X)-K,0
\right\},
\]
and we wish to compute $\expec \phi(X)$.

To use \textsc{Convex AdaMC}, we choose the exponential family
\[
f_\theta(x)=\frac{1}{(2\pi)^{k/2}}e^{-\|x-\theta\|_2^2}
=\exp\left(\theta^Tx-\frac{1}{2}\|\theta\|_2^2\right)
\exp\left(-\frac{1}{2}\|x\|_2^2\right)/(2\pi)^{k/2}
,
\]
where $\theta\in \reals^k$ and $\Theta\in [-0.5,0.5]^k$.
In other words, $X\sim f_\theta$
contains independent standard normals with mean shifted by $\theta$.
So we have
\[
\phi(X)\frac{f(X)}{f_\theta(X)}
=\phi(X)\exp\left(\frac{1}{2}\|\theta\|_2^2-X^T\theta\right)
\]
and
\[
\nabla A(\theta)=\theta
\qquad
T(X)=X.
\]

We run \text{Convex AdaMC} with the parameters
$S^0=K=50$, $r=0.05$, $\sigma=0.3$,
$T=1.0$, $k=64$, $C=0.01$, and $\theta_1=0$
for $10^7$ iterations.
Figure~\ref{fig-asian} shows the shifting at the end of the algorithm
and the asset price estimate is $4.02$.
\begin{figure}[h!]
\vskip -0.15in
\hskip -0.5in
\begin{psfrags}
\psfrag{entries of theta}{\raisebox{-1ex}{\footnotesize Entries of $\theta$}}
\includegraphics[width=1.15\textwidth]{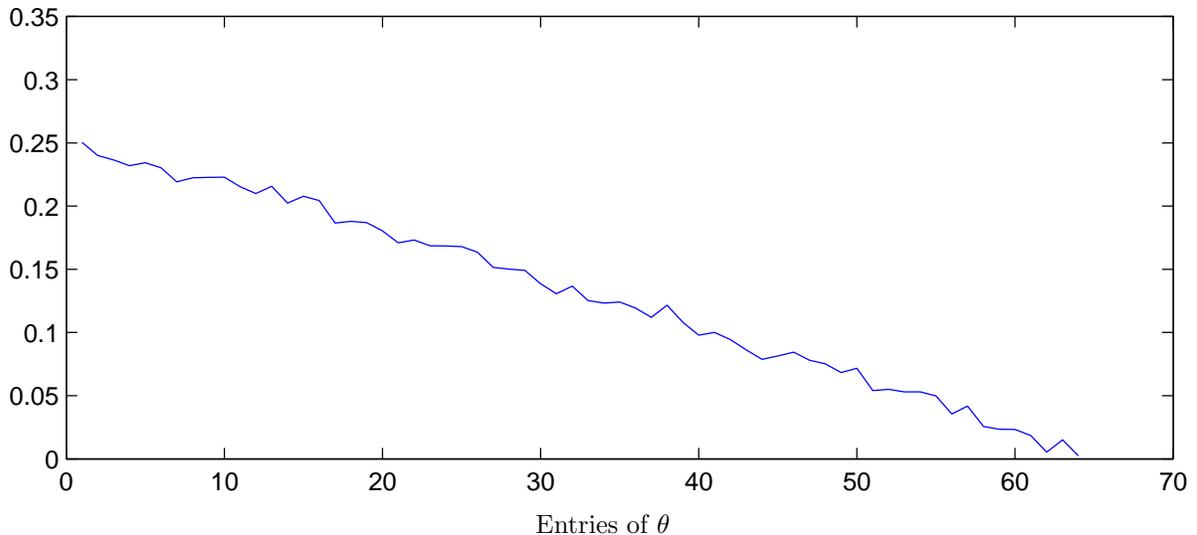}
\end{psfrags}
\vskip -0.05in
\caption{Importance sampling parameter $\theta$
for the Asian option pricing problem
after $10^7$ iterations.}
\label{fig-asian}
\end{figure}

\section{Remarks, extensions, and variations}
\label{s-conclusion}
An iteration of \textsc{Convex AdaMC} is simultaneously an iteration of a convex optimization problem
and an iteration of importance sampling.
Because of this fact, each iteration of the method is computationally efficient,
and we can prove convergence of the variance (and of course the estimator)
as a function of the iteration count.

Some previous work on adaptive importance sampling
have used stochastic gradient descent or similar stochastic approximation algorithms
without a setup to make the variance convex \cite{alqaq1995, arouna2003, arouna2004, egloff2010}.
While these methods are applicable to a more general class of candiate sampling distributions,
they have little theoretical guarantees on the variances of the estimators;
This is not surprising as in general with nonconvex optimization problems,
it is difficult to prove
anything beyond mere convergence to a stationary point,
such as a rate of convergence or convergence to the global optimum.

Other previous work on adaptive importance sampling
solves an optimization subproblem to update the sampling parameter each time,
either with an off-the-shelf deterministic optimization algorithm
or, especially in the case of the cross-entropy method, by focusing on special cases with analytic solutions
\cite{oh1992, devetsikiotis1993a, devetsikiotis1993b, lieber1997, rubinstein1997, rubinstein1999, deboer2004, deboer2005, pennanen2006, douc2007, cappe2008, corneut2012, he2014}.
While some these methods do exploit convexity to establish that the subproblems can be solved efficiently,
these subproblems and the storage requirement to represent these subproblems grow in size with the number of iterations.
One could loosely argue that the inefficiency is a consequence of separating the optimization and the importance sampling.


We point out two straightforward generalizations that we omitted for the sake of simplicity.
One is that when the nominal and importance distributions have densities with respect to
any general measure, not the Lebesgue measure as assumed in our exposition, the same results apply.

Another is to adaptively minimize the R\'{e}nyi generalized divergence with parameter $\alpha\ge 1$
of the sampling distribution to the 
``optimal'' sampling distribution $|\phi(x)|f(x)$ \cite{renyi1961, mcLeish2010}.
What we did, minimizing the variance of the estimator, is the special with $\alpha=2$.
When $\alpha=1$, the R\'{e}nyi generalized divergence becomes the cross entropy
and we get a method similar to the cross-entropy method \cite{rubinstein1999,deboer2005}.
(The R\'{e}nyi generalized divergence is convex for $\alpha\ge 1$.)

There are other not-so-straightforward generalizations worth pursuing.
One is to try other stochastic optimization methods.
In this paper, we used the most common and simplest stochastic optimization method,
stochastic gradient descent with step size $C/\sqrt{n}$.
However, there are many other methods to solve a given stochastic optimization problem,
and these other methods could perform better under certain assumptions.

Another would be a different weighting scheme.
In \textsc{Convex AdaMC},
we add a sequence of unbiased estimators with varying variance,
which, loosely speaking, is decreasing in expectation.
If we knew these variances in advance, then we can easily compute
the optimal weighting, which is not a uniform weighting.
Although we don't know the variances in advance,
it would be interesting to know if there is a better weighting
or to characterize the optimality of the uniform weighting
in the spirit of Theorem~4 of \cite{owen1999}.

Finally, it would be most interesting to understand
\textsc{Convex AdaMC}'s theoretical and empirical
performance when used in conjunction with
other variance reduction techniques such as control variates or
mixture importance sampling.

\subsection*{Acknowledgement}
We thank Art B. Owen for helpful discussions and many detailed suggestions.
This research was supported by the Simons Foundation and DARPA X-DATA.

\section{Appendix}
The following lemma, like Lemma~\ref{lm-diff}
follows from Theorem 2.7.1 of \cite{lehmann2005}.
\begin{lm}
\label{lm-k}
$A$, $V$, and $K$ are infinitely differentiable 
and all derivatives can be evaluated under their integrals
on the interiors of their respective domains.
\end{lm}

We are now ready to prove the part of
the proof of Theorem~\ref{thm-main} we omitted.
\begin{proof}[Proof of $G<\infty$ in Theorem~\ref{thm-main}]
For any $i\in \{1,2,\ldots,p\}$, 
\[
\frac{\partial}{\partial \theta_i} K(\theta)=
3\int
\left(\frac{\partial}{\partial \theta_i} A(\theta)-T_i(x)\right)\frac{\phi^4(x)f^4(x)}{f_\theta^3(x)}\;dx
=3K(\theta)
\frac{\partial}{\partial \theta_i}
 A(\theta)-
3
\int T_i(x)\frac{\phi^4(x)f^4(x)}{f_\theta^3(x)}\;dx
\]
exists and is and continuous on $\intr\left\{\theta\mid K(\theta)<\infty\right\}$
by Lemma~\ref{lm-k}.
As we already know the first term is continuous by Lemma~\ref{lm-k},
this tells us the second term
is continuous.

Repeating this, we have
\begin{align*}
\frac{\partial^2}{\partial \theta_i^2} K(\theta)
=
\int
\left(3\frac{\partial^2}{\partial \theta_i^2} A(\theta)+
9\left(\frac{\partial}{\partial \theta_i}  A(\theta)\right)^2
-18T_i(x)
\frac{\partial}{\partial \theta_i}
A(\theta)+9T_i^2(x)\right)
\frac{\phi^4(x)f^4(x)}{f_\theta^3(x)}\;dx.
\end{align*}
We know the first 3 terms are continuous 
from what we just proved and Lemma~\ref{lm-k}.
So we conclude that
\[
\int
T_i^2(x)
\frac{\phi^4(x)f^4(x)}{f_\theta^3(x)}\;dx
\]
is a continuous function of $\theta$ on $\intr\left\{\theta\mid K(\theta)<\infty\right\}$.

Finally, we conclude that
\begin{align*}
&\expec_{X\sim f_\theta}
\left\|(\nabla A(\theta)-T(X))\frac{\phi^2(X)f^2(X)}{f_\theta^2(X)}\right\|_2^2\\
&=
\|\nabla A(\theta)\|_2^2K(\theta)
-2\nabla A(\theta)^T
\int
T(X)
\frac{\phi^4(X)f^4(X)}{f_\theta^3(X)}
\;dx
+
\int
\|T(X)\|_2^2
\frac{\phi^4(X)f^4(X)}{f_\theta^3(X)}
\;dx
\end{align*}
is a continuous function on the compact set $\Theta$
and therefore the supremum, $G^2$, is finite.
\end{proof}

Finally, we prove the CLT.
\begin{proof}[Proof of Theorem~\ref{thm-clt}]
First define
\[
Y_{ni}=
\left\{
\begin{array}{ll}
\frac{1}{\sqrt{n}}\left(\frac{\phi(X_i)f(X_i)}{f_{\theta_i}(X_i)}-I\right)&\text{for }i\le n\\
0&\text{otherwise}
\end{array}
\right.
\]
and
\[
J_{nm}=\sum^m_{i=1}Y_{ni}.
\]
Also define the $\sigma$\nobreakdash-algebras
\[
\mathcal{G}_m=\sigma(\theta_1,\theta_2,\ldots, \theta_{m+1},X_1,X_2,\ldots,X_m)
\]
for all $m$. Then for any given $n$, the process
$J_{n1},J_{n2},\ldots$ is a martingale with respect to $\mathcal{G}_1,\mathcal{G}_2,\ldots$
and to we have to prove
\[
J_{nn}=\sum^n_{i=1}Y_{ni}=
\sum^\infty_{i=1}Y_{ni}\stackrel{\mathcal{D}}{\rightarrow}
\mathcal{N}(0,V^\star).
\]

Define
\[
\sigma_{ni}^2=
\expec\left[
Y_{ni}^2
\mid \mathcal{G}_{i-1}
\right]=
\left\{
\begin{array}{ll}
\frac{1}{n}V(\theta_i)& \text{for }i\le n\\
0&\text{otherwise}.
\end{array}
\right.
\]
Then a form of the Martingale CLT, \emph{c.f.}, Theorem 35.12 of \cite{billingsley1995},
states that if
\[
\sum^n_{i=1}\sigma_{ni}^2\stackrel{\mathcal{P}}{\rightarrow}V^\star
\]
and
\[
\sum^n_{i=1}\expec Y^2_{ni}I_{\{|Y_{ni}|\ge \varepsilon\}}
\rightarrow 0
\]
for each $\varepsilon>0$, then $J_{nn}\stackrel{\mathcal{D}}{\rightarrow}\mathcal{N}(0,V^\star)$.

Since
\[
\sum^n_{i=1}\sigma^2_{ni}
=
\frac{1}{n}\sum^n_{i=1}V(\theta_i),
\]
and since, by Theorem~\ref{thm-main}, we have
\[
\frac{1}{n}\sum^n_{i=1}\left(\expec V(\theta_i)-V^\star\right)
=
\expec\left|\frac{1}{n}\sum^n_{i=1} V(\theta_i)-V^\star\right|
=\mathcal{O}(1/\sqrt{n})
\rightarrow 0,
\]
\ie, $\sum^n_{i=1}\sigma^2_{ni}$ converges to $V^\star$ in $L^1$,
we have
\[
\sum^n_{i=1}\sigma^2_{ni}
\stackrel{\mathcal{P}}{\rightarrow}V^*.
\]

Finally, since $\Theta\subseteq\left\{\theta\mid K(\theta)<\infty\right\}$
is a compact set and $K(\theta)$ is a continuous function 
by Lemma~\ref{lm-k},
we have
\[
B=\sup_{\theta\in \Theta}K(\theta)<\infty
\]
and we conclude
\BEAS
\sum^n_{i=1}\expec Y^2_{ni}I_{\{|Y_{ni}|\ge \varepsilon\}}
&=&\frac{1}{n}
\sum^n_{i=1}\expec 
\frac{\phi^2(X_i)f^2(X_i)}{f_{\theta_i}^2(X_i)}
I_{\left\{\phi^2(X_i)f^2(X_i)/f_{\theta_i}^2(X_i)\ge n\varepsilon^2\right\}}\\
&\le& \frac{1}{n^2\varepsilon^2}
\sum^n_{i=1}\expec 
\frac{\phi^4(X_i)f^4(X_i)}{f_{\theta_i}^4(X_i)}\le \frac{B}{n\varepsilon^2}\rightarrow 0.
\EEAS
Since this proves the conditions we need, applying the martingale CLT completes the proof.
\end{proof}

\bibliography{bibliography}{}

\begin{thebibliography}{10}

\bibitem{alqaq1995}
W.~A. Al-Qaq, M.~Devetsikiotis, and J.~K. Townsend.
\newblock Stochastic gradient optimization of importance sampling for the
  efficient simulation of digital communication systems.
\newblock {\em {IEEE} Transactions on Communications}, 43(12):2975--2985, 1995.

\bibitem{arouna2003}
B.~Arouna.
\newblock {Robbins-Monro} algorithms and variance reduction in finance.
\newblock {\em The Journal of Computational Finance}, 7(2):35--61, 2003.

\bibitem{arouna2004}
B.~Arouna.
\newblock Adaptive {Monte Carlo} method, a variance reduction technique.
\newblock {\em Monte Carlo Methods and Applications}, 10(1):1--24, 2004.

\bibitem{billingsley1995}
P.~Billingsley.
\newblock {\em Probability and Measure}.
\newblock Wiley, third edition, 1995.

\bibitem{boyd2004}
S.~Boyd and L.~Vandenberghe.
\newblock {\em Convex Optimization}.
\newblock Cambridge University Press, 2004.

\bibitem{cappe2008}
O.~Capp\'e, R.~Douc, A.~Guillin, J.-M. Marin, and C.~P. Robert.
\newblock Adaptive importance sampling in general mixture classes.
\newblock {\em Statistics and Computing}, 18(4):447--459, 2008.

\bibitem{corneut2012}
J.-M. Corneut, J.-M. Marin, A.~Mira, and C.~P. Robert.
\newblock Adaptive multiple importance sampling.
\newblock {\em Scandinavian Journal of Statistics}, 39(4):798--812, 2012.

\bibitem{deboer2005}
P.-T. de~Boer, D.~P. Kroese, S.~Mannor, and R.~Y. Rubinstein.
\newblock A tutorial on the cross-entropy method.
\newblock {\em Annals of Operations Research}, 134(1):19--67, 2005.

\bibitem{deboer2004}
P.-T. de~Boer, D.~P. Kroese, and R.~Y. Rubinstein.
\newblock A fast cross-entropy method for estimating buffer overflows in
  queueing networks.
\newblock {\em Management Science}, 50(7):883--895, 2004.

\bibitem{devetsikiotis1993a}
M.~Devetsikiotis and J.~K. Townsend.
\newblock An algorithmic approach to the optimization of importance sampling
  parameters in digital communication system simulation.
\newblock {\em IEEE Transactions on Communications}, 41(10):1464--1473, 1993.

\bibitem{devetsikiotis1993b}
M.~Devetsikiotis and J.~K. Townsend.
\newblock Statistical optimization of dynamic importance sampling parameters
  for efficient simulation of communication networks.
\newblock {\em IEEE/ACM Transactions on Networking}, 1(3):293--305, 1993.

\bibitem{douc2007}
R.~Douc, R.~Guillin, J.-M. Marin, and C.~P. Robert.
\newblock Convergence of adaptive mixtures of importance sampling schemes.
\newblock {\em The Annals of Statistics}, 35(1):420--448, 2007.

\bibitem{egloff2010}
D.~Egloff and M.~Leippold.
\newblock Quantile estimation with adaptive importance sampling.
\newblock {\em The Annals of Statistics}, 38(2):1244--1278, 2010.

\bibitem{glasserman1999}
P.~Glasserman, P.~Heidelberger, and P.~Shahabuddin.
\newblock Asymptotically optimal importance sampling and stratification for
  pricing path-dependent options.
\newblock {\em Mathematical Finance}, 9(2):117--152, 1999.

\bibitem{he2014}
H.~Y. He and A.~B. Owen.
\newblock Optimal mixture weights in multiple importance sampling.
\newblock 2014.

\bibitem{hiriart2001}
J.-B. Hiriart-Urruty and C.~Lemar\'echal.
\newblock {\em Fundamentals of Convex Analysis}.
\newblock Springer, 2001.

\bibitem{kushner2003}
H.~J. Kushner and G.~G. Yin.
\newblock {\em Stochastic Approximation and Recursive Algorithms and
  Applications}.
\newblock Springer, 2003.

\bibitem{lehmann2005}
E.~L. Lehmann and J.~P. Romano.
\newblock {\em Testing Statistical Hypotheses}.
\newblock Springer, third edition, 2005.

\bibitem{lieber1997}
D.~Lieber, R.~Y. Rubinstein, and D.~Elmakis.
\newblock Quick estimation of rare events in stochastic networks.
\newblock {\em IEEE Transactions on Reliability}, 46(2):254--265, 1997.

\bibitem{madras2002}
N.~Madras.
\newblock {\em Lectures on Monte Carlo Methods}.
\newblock American Mathematical Society, 2002.

\bibitem{mcLeish2010}
D.~L. McLeish.
\newblock Bounded relative error importance sampling and rare event simulation.
\newblock {\em ASTIN Bulletin}, 40(1), 2010.

\bibitem{nesterov2004}
Y.~Nesterov.
\newblock {\em Introductory Lectures on Convex Optimization: A Basic Course}.
\newblock Kluwer Academic Publishers, 2004.

\bibitem{oh1992}
M.-S. Oh and J.~O. Berger.
\newblock Adaptive importance sampling in {Monte Carlo} integration.
\newblock {\em Journal of Statistical Computation and Simulation}, 41:143--168,
  1992.

\bibitem{owen2013}
A.~B. Owen.
\newblock {\em Monte Carlo theory, methods and examples}.
\newblock 2013.

\bibitem{owen1999}
A.~B. Owen and Y.~Zhou.
\newblock Adaptive importance sampling by mixtures of products of beta
  distributions.
\newblock Technical report, Stanford University, 1999.

\bibitem{pennanen2006}
T.~Pennanen and M.~Koivu.
\newblock An adaptive importance sampling technique.
\newblock In H.~Niederreiter and D.~Talay, editors, {\em Monte Carlo and
  Quasi-Monte Carlo Methods 2004}, pages 443--455. Springer, 2006.

\bibitem{polyak1987}
B.~T. Polyak.
\newblock {\em Introduction to Optimization}.
\newblock Optimization Software, Inc., 1987.

\bibitem{renyi1961}
A.~R\'{e}nyi.
\newblock On measures of entropy and information.
\newblock In {\em Proceedings of the Fourth Berkeley Symposium on Mathematical
  Statistics and Probability, Volume 1: Contributions to the Theory of
  Statistics}. University of California Press, 1961.

\bibitem{robbins1951}
H.~Robbins and S.~Monro.
\newblock A stochastic approximation method.
\newblock {\em The Annals of Mathematical Statistics}, 22(3):400--407, 1951.

\bibitem{robert2004}
C.~Robert and G.~Casella.
\newblock {\em Monte Carlo Statistical Methods}.
\newblock Springer, 2004.

\bibitem{rockafellar1970}
R.~T. Rockafellar.
\newblock {\em Convex Analysis}.
\newblock Princeton University Press, 1970.

\bibitem{rubinstein1997}
R.~Y. Rubinstein.
\newblock Optimization of computer simulation models with rare events.
\newblock {\em European Journal of Operations Research}, 99:89--112, 1997.

\bibitem{rubinstein1999}
R.~Y. Rubinstein.
\newblock The cross-entropy method for combinatorial and continuous
  optimization.
\newblock {\em Methodology And Computing In Applied Probability},
  1(2):127--190, 1999.

\bibitem{sadowsky1990}
J.~S. Sadowsky and J.~A. Bucklew.
\newblock On large deviations theory and asymptotically efficient monte carlo
  estimation.
\newblock {\em {IEEE} Transactions on Information Theory}, 36(3):579--588,
  1990.

\bibitem{shor1985}
N.~Z. Shor.
\newblock {\em Minimization Methods for Non-Differentiable Functions}.
\newblock Springer, 1985.

\bibitem{smith1997}
P.~J. Smith, M.~Shafi, and H.~Gao.
\newblock Quick simulation: {A} review of importance sampling techniques in
  communication systems.
\newblock {\em {IEEE} Journal on Selected Areas in Communications},
  15(4):597--613, 1997.

\bibitem{srinivasan2002}
R.~Srinivasan.
\newblock {\em Importance Sampling: Applications in Communications and
  Detection}.
\newblock Springer, 2002.

\end{thebibliography}
\bibliographystyle{plain}

\end{document}